\documentclass[12pt,a4paper]{article}
\usepackage{amsmath,amssymb,amsfonts,amsthm}
\usepackage{graphicx}
\usepackage{natbib}
\usepackage{float}
\usepackage{dsfont} 
\numberwithin{equation}{section}
\numberwithin{figure}{section}
\numberwithin{table}{section}
\newtheorem{definition}{Definition}[section]
\newtheorem{proposition}{Proposition}[section]

\newtheorem{algorithm}{Algorithm}[section]
\newtheorem{theorem}{Theorem}[section]
\newcommand{\1}{\mathds{1}}
\newcommand{\Z}{\mathbb{Z}}
\newcommand{\bhline}[1]{\noalign{\hrule height #1}}
\begin{document}
\title{Recursive computation for evaluating the exact $p$-values of temporal and spatial scan statistics}
\author{Satoshi Kuriki\thanks{
The Institute of Statistical Mathematics,
10-3 Midoricho, Tachikawa, Tokyo 190-8562, Japan,
Email: \texttt{kuriki@ism.ac.jp}
(corresponding author)
},
\ \ %
Kunihiko Takahashi\thanks{
Department of Biostatistics, Nagoya University Graduate School of Medicine, 
65 Tsurumi-cho, Showa-ku, Nagoya 466-8550, Japan,
Email: \texttt{kunihiko@med.nagoya-u.ac.jp}
},
\ \ %
Hisayuki Hara\thanks{
Faculty of Economics, Niigata University,
8050 Ikarashi 2-no-cho Nishi-ku, Niigata 950-2181, Japan,
Email: \texttt{hara@econ.niigata-u.ac.jp}
}
}
\date{}
\maketitle
\begin{abstract}
Let $V$ be a finite set of indices, and let $B_i$, $i=1,\ldots,m$, be subsets of $V$ such that $V=\bigcup_{i=1}^{m}B_i$. 
Let $X_i$, $i\in V$, be independent random variables, and let $X_{B_i}=(X_j)_{j\in B_i}$.
In this paper, we propose a recursive computation method to calculate the conditional expectation $E\bigl[\prod_{i=1}^m\chi_i(X_{B_i}) \,|\, N\bigr]$ with $N=\sum_{i\in V}X_i$ given, where $\chi_i$ is an arbitrary function.
Our method is based on the recursive summation/integration technique using the Markov property in statistics.
To extract the Markov property, we define an undirected graph whose cliques are $B_j$, and obtain its chordal extension, from which we present the expressions of the recursive formula.
This methodology works for a class of distributions including the Poisson distribution (that is, the conditional distribution is the multinomial).
This problem is motivated from the evaluation of the multiplicity-adjusted $p$-value of scan statistics in spatial epidemiology.
As an illustration of the approach, we present the real data analyses to detect temporal and spatial clustering.

\end{abstract}

\section{Introduction}
\label{sec:introduction}
Let
\[
 X_V=(X_i)_{i\in V}=(X_1,\ldots,X_n), \quad V=\{1,\ldots,n\},
\]
be a random vector whose index set is $V$.
Throughout the paper, we use the convention that $X_Z=(X_i)_{i\in Z}$ when $Z$ is a set of indices.
Suppose that $X_i$ are distributed independently according to the Poisson distribution and consider the distribution of $X_V$ when $\sum_{i\in V}X_i=N$ is given.
That is, $X_V$ follows the multinomial distribution with probability $p_V=(p_i)_{i\in V}$ and total sum $N$:
\[
 X_V |_{\sum_{i\in V}X_i=N} \,\sim \mathrm{Mult}(N;p_V).
\]
Let $Z_1,\ldots,Z_m$ be subsets of $V$ satisfying $V=\bigcup_{i=1}^m Z_i$.
The main technical result of the paper is to provide an algorithm to evaluate the conditional expectation
\begin{equation}
\label{expectation}
 E\Bigl[{\prod}_{i=1}^m \chi_i(X_{Z_i}) \,|\, N \Bigr],
\end{equation}
where $\chi_i(X_{Z_i})$ is an arbitrary function of $X_{Z_i}$.
We also consider a generalization where $X_i$ are distributed as a class of distributions including the multinomial distribution.

This problem arises from the evaluation of the multiplicity-adjusted $p$-value of scan statistics.
We begin with stating the scan statistics in spatial epidemiology, which is a typical example in this framework.

In a certain country, there are $n$ districts.
Let $V=\{1,\ldots,n\}$ be the set of districts.
For each district $i\in V$, we suppose that the number $X_i$ for event we are focusing on (e.g., number of patients with some disease) as well as its expected frequency $\lambda_i$ estimated from historical data are available.
$X_i$ is assumed to be distributed according to the Poisson distribution with parameter $\theta_i\lambda_i$, $\mathrm{Po}(\theta_i\lambda_i)$, independently for $i\in V$.
The parameter $\theta_i$ is known as the standardized mortality ratio (SMR).
Figure \ref{fig:yamagata-smr} depicts an example of a choropleth map of SMRs.
For 44 districts, we indicate the values of the SMRs with different colors.

A set of adjacent districts with SMRs higher than other areas is called a \textit{disease cluster\/}.
The detection of such disease clusters is a major interest in spatial epidemiology.
To detect a disease cluster, we settle a family of subsets
\[
 \mathcal{Z} = \{ Z_1,\ldots,Z_m \}, \quad Z_i\subset V,
\]
as candidates of a disease cluster, and define a scan statistic $\varphi_{Z_i}(X_{Z_i})$ for each $Z_i\in\mathcal{Z}$.
$Z_i$ is called the \textit{scan window\/}.
The choice of the scan windows is an important research topic in spatial epidemiology
\citep{Kulldorff06}.
When $\varphi_{Z_i}(X_{Z_i})$ is larger than a threshold, say $c$, we declare that $Z_i$ is a disease cluster.
As such a scan statistic, \cite{Kulldorff97} proposed the use of the likelihood ratio test (LRT) statistic $\varphi_Z(X_Z)$ for the null hypothesis
$H_0\,:\,\theta_i \equiv \theta_{0}$ (constant) against the alternative
\[
 H_Z : \ \theta_i=\theta_{Z} \ (i\in Z),\ \ \theta_i=\theta_{\overline Z} \ (i\notin Z) \ \ \mbox{such that}\ \theta_{Z} > \theta_{\overline Z}
\]
under the conditional distribution with $N=\sum_{i\in V}X_i$ given.
The conditional inference (inference under the conditional distribution) leads to a similar test independent of the nuisance parameter $\theta_{0}$.
The expression of $\varphi_Z(X_Z)$ is given in Section \ref{subsec:scan}.
When the hypothesis $H_Z$ holds, the disease cluster $Z$ is called a hotspot.

This is a typical problem of multiple comparisons.
The $p$-value to assess the significance should be adjusted to incorporate the multiplicity effect.
One method is to define the $p$-value from the distribution of the maximum $\max_{Z \in \mathcal{Z}} \varphi_Z(X_Z)$ under $H_0$:
\begin{equation}
 P\Bigl(\max_{Z \in \mathcal{Z}} \varphi_Z(X_Z) \le c \,|\, N \Bigr)
= E\Bigl[{\prod}_{Z\in\mathcal{Z}} \chi_Z(X_Z) \,|\, N \Bigr],
\label{conditional}
\end{equation}
where $\chi_Z(X_Z) = \1{\{ \varphi_Z(X_Z) \le c \}}$.
This is of the form of (\ref{expectation}).
Note that when $N$ is given, the distribution of $X_V$ is
$\mathrm{Mult}(N; p_V)$, where $p_i=\lambda_i/\sum_{i\in V}\lambda_i$.

In spatial epidemiology, the $p$-value is usually estimated using Monte Carlo simulation.
Although this is convenient and practical in most cases, when the true $p$-value is very small, it is difficult to obtain a precise value even when the number of random numbers in the Monte Carlo is large.
Therefore, we have good reason to conduct exact computation according to the definition (\ref{conditional}) by enumerating all possibilities.
However, this is generally difficult because of the computational complexity. 

In the area of multiple testing comparisons, many techniques to reduce computational time have been proposed.
For example, \cite{Worsley86} demonstrated that in a change point problem,
the computational time for the distribution of the maximum could be reduced using the Markov property among statistics.
See also \cite{Kuriki-etal02}, \cite{Hirotsu-Marumo02} and references therein.
In this paper, we develop a similar computation technique by taking advantage of the Markov structure among scan statistics.
The proposed method is based on the theory of a chordal graph, which is the foundation of the theory of graphical models \citep{Lauritzen96}.
The chordal graph theory provides rich tools, in not only statistics but also many fields of mathematical science.
In particular, in numerical analysis, this is a major tool used to conduct large-scale matrix computation (e.g., \cite{Rose71},\cite{Yamashita08}).
We also apply the chordal graph theory to retrieve the Markov structure to reduce the computational time by using the recursive summation/integration technique.
Our technique is then similar to those used in the efficient computation of maximum likelihood estimator of graphical models
(e.g., \cite{Badsberg-Malvestuto}, \cite{Hara-Takemura10}, \cite{XuGuoHe11}, \cite{XuGuoTang12}, \cite{XuGuoTang14}). 

The remainder of the paper is organized as follows.
Section \ref{sec:recursive} provides the recurrence computational formula in the multinomial distribution under the assumption that the running intersection property holds.
We evaluate the computational complexity,
and show that the recurrence computation methodology works for a class of distributions including the Poisson distribution (that is, the conditional distribution is the multinomial).
Section \ref{sec:markov} proposes a method to detect the Markov property,
and Section \ref{sec:examples} presents illustrative real data analyses to detect temporal and spatial clustering.
Section \ref{sec:summary} briefly summarizes our results and discusses further research.

\section{Recursive computation of conditional expectations}
\label{sec:recursive}

In this section, we provide an algorithm to compute the expectation
(\ref{expectation}) when the sequence $Z_1,\ldots,Z_m$ of subsets of $V$ 
has a nice property,
which is called the running intersection property given in Definition \ref{def:rip}.
We will consider the general case in the next section.
In this section, we use the symbol $B_i$ instead of $Z_i$.

\subsection{The case where $m=2$}

We start with case $m=2$.
For $1<l_1\le l_2<n$, let $B_1=\{1,\ldots,l_2\}$ and $B_2=\{l_1+1,\ldots,n\}$.
Suppose that
\[
 X_{B_1\cup B_2} = (X_1,\ldots,X_n) \sim \mathrm{Mult}(N;p_{B_1\cup B_2})
\]
is a random vector distributed according to the multinomial distribution with summation $N$ and probability $p_{B_1\cup B_2}=(p_1,\ldots,p_n)$.
We consider the evaluation of the expectation
\begin{equation}
 E[\chi_1(X_{B_1}) \chi_2(X_{B_2})]
\label{naive}
\end{equation}
exactly, where $\chi_1$ and $\chi_2$ are arbitrary functions.

Obviously, $X_{B_1}$ and $X_{B_2}$ are not independent;
there is an overlap $X_{C_1}=(X_{l_{1+1}},\ldots,X_{l_2})$ unless $C_1=B_1\cap B_2$ is empty.
Moreover, there is a restriction that
\[
\sum_{i\in B_1}X_i + \sum_{i\in B_2}X_i - \sum_{i\in C_1}X_i = N.
\]

Instead of the problem of evaluating the expectation, by a change of viewpoint,
we first consider the problem of generating random numbers 
$X_{B_1\cup B_2}=(X_{R_1},X_{R_2})$, where $R_1=B_1\setminus C_1$, $R_2=B_2$.
$X_{B_1\cup B_2}$ can be generated according to the following three steps:
\begin{align}
 (M_2,M_1)|_N &\sim \mathrm{Mult}\Bigl(N;\Bigl(\sum_{i\in R_2} p_i,\sum_{i\in R_1} p_i\Bigr)\Bigr),
\label{M2M1} \\
 X_{R_2}|_{M_2} &\sim \mathrm{Mult}\Bigl(M_2;p_{R_2}/\sum_{i\in R_2} p_i\Bigr),
\label{XR2} \\
 X_{R_1}|_{M_1} &\sim \mathrm{Mult}\Bigl(M_1;p_{R_1}/\sum_{i\in R_1} p_i\Bigr),
\label{XR1}
\end{align}
where $X_{R_2}$ and $X_{R_1}$ are independent given $(M_2,M_1)$.  
Correspondingly, we divide the expectation in (\ref{naive}) into three parts as
\begin{align}
E[\chi_1(X_{B_1}) \chi_2(X_{B_2})]
&= E^{(M_2,M_1)|N}\bigl[ E^{X_{R_2}|M_2}\bigl[ E^{X_{R_1}|M_1}[\chi_1(X_{B_1}) \chi_2(X_{B_2})] \bigr]\bigr]  \nonumber \\
&= E^{(M_2,M_1)|N}\bigl[ E^{X_{R_2}|M_2}\bigl[ \chi_2(X_{B_2}) E^{X_{R_1}|M_1}[\chi_1(X_{B_1}) ] \bigr]\bigr]  \nonumber \\
&= E^{(M_2,M_1)|N}\bigl[ E^{X_{R_2}|M_2}[ \chi_2(X_{B_2}) \xi_1(M_1,X_{C_1})] \bigr],
\label{E}
\end{align}
where
\begin{align}
\xi_1(M_1,X_{C_1})
&= E^{X_{R_1}|M_1}[\chi_1(X_{B_1})] \nonumber \\
&= E^{X_{R_1}|M_1}[\chi_1(X_{C_1},X_{R_1})].
\label{xi}
\end{align}

The procedure for the numerical computation of (\ref{E}) is as follows.
(i) For possible values of $M_1$ and $X_{C_1}$, compute $\xi_1(M_1,X_{C_1})$ in (\ref{xi}).
Here, the expectation is taken over $X_{R_1}$ according to (\ref{XR1}) with $X_{C_1}$ fixed.
The results are stored in memory as a tabulation.
(ii) Compute $E[\chi_1(X_{B_1}) \chi_2(X_{B_2})]$ according to (\ref{E}).
Here, $M_2$, $M_1$, and $X_{R_2}$ are random variables having distributions (\ref{M2M1}) and (\ref{XR2}). 

This technique substantially reduces the computational cost.
To see this, we enumerate the number of required summations in detail.
$M_1$ takes the values $0\le M_1\le N$ and $X_{C_1}$ takes the values of all nonnegative integer vectors whose sum is less than or equal to $M_2=N-M_1$.
In the expectation $E^{X_{R_1|M_1}}[\cdot]$ in (\ref{xi}),
the variable $X_{R_1}$ runs over all nonnegative vectors whose sum is $M_1$.
Because
\[
 \sharp\Bigl\{ (x_1,\ldots,x_d)\in\Z^d \mid x_i\ge 0,\ {\sum}_{i=1}^d x_i=n \Bigr\} = {n+d-1 \choose d-1},
\]
\[
 \sharp\Bigl\{ (x_1,\ldots,x_d)\in\Z^d \mid x_i\ge 0,\ {\sum}_{i=1}^d x_i\le n \Bigr\} = {n+d \choose d},
\]
we see that the number of summations to prepare the table $\xi_1(M_1,X_{C_1})$ is
\[
\sum_{M_1=0}^N {N-M_1+|C_1| \choose |C_1|} {M_1+|R_1|-1 \choose |R_1|-1}
= {N+|C_1|+|R_1| \choose |C_1|+|R_1|} = {N+|B_1| \choose |B_1|}.
\]
In the expectations $E^{(M_2,M_1)|N}[ E^{X_{R_2}|M_2}[\cdot]]$ in (\ref{E}), the variable $X_{R_2}$ runs over all nonnegative vectors whose sum is $M_2$, and $M_2+M_1=N$.
Hence, the number of summations for computing (\ref{E}) is
\[
 \sum_{M_2=0}^N {M_2+|R_2|-1 \choose |R_2|-1} = {N+|R_2| \choose |R_2|} = {N+|B_2| \choose |B_2|}.
\]
Therefore, the number of summations is
\[
 {N+|B_1| \choose |B_1|} + {N+|B_2| \choose |B_2|} = O(N^{\max(|B_1|,|B_2|)}).
\]
On the other hand, when we do not use this technique, the number of summations is
\[
 {N+|B_1\cup B_2|-1 \choose |B_1\cup B_2|-1} = O(N^{|B_1\cup B_2|-1}),
\]
which is of larger order than $O(N^{\max(|B_1|,|B_2|)})$.

\subsection{The case of general $m$}

This technique for reducing computational time is available in a general setting.
To describe it, we introduce several notions from the theory of chordal graphs and graphical models \citep{Blair-Peyton93, Lauritzen96}.
\begin{definition}
\label{def:rip}
A sequence of sets $B_1,\ldots,B_m$ is said to have the running intersection property (RIP) if for each $1\le i\le m-1$, there is a $k(i)>i$ such that
\begin{equation}
\label{rip}
 B_i \cap \Bigl( {\bigcup}_{j>i} B_j \Bigr) = B_i \cap B_{k(i)}.
\end{equation}
\end{definition}

Throughout this section, we suppose that sequence $B_1,\ldots,B_m$ satisfies the running intersection property.
Note that the indices $i$ of $B_i$ in Definition \ref{def:rip} are reversely ordered from the conventional definition.

The function $k(\cdot)$ defines a directed graph $(V,E)$,
where $V=\{1,\ldots,m\}$ and $E=\{(i,k(i)) \mid i\in V \}$.
The running intersection property is explained in detail in Section \ref{sec:markov}.

For $i,j\in V$, write $i\preceq j$ iff $j=i$ or
\[
 j=\underbrace{k(k(\cdots(k}_{h\,\mathrm{times}}(i))\cdots)) \ \ \mbox{for some $h$}.
\]

Let $R_i=B_i\setminus B_{k(i)}$ ($i<m$) and $R_m=B_m$, where
$R_i$ are residual sets.
The disjoint union is denoted by $\sqcup$.
The proposition below follows from the running intersection property.
\begin{proposition}
\[
 V = \bigsqcup_{i=1}^m R_i.
\]
\end{proposition}
\begin{proof}
Obviously $V=\bigcup_{i=1}^m B_i=\bigcup_{i=1}^m R_i$.
We prove $R_i\cap R_j=\emptyset$ for $i\ne j$.
Suppose that $R_i\cap R_j\ne\emptyset$ for $i<j$.
Let $x\in R_i\cap R_j$.
Since $x\in B_i,B_j$, $x\in B_i\cap(\bigcup_{j>i} B_j)=B_i\cap B_{k(i)}$.
On the other hand, $x\in R_i$ implies $x\notin B_{k(i)}$.
This is a contradiction.
\end{proof}

Write $C_i=B_i\cap B_{k(i)}$ ($i=1,\ldots,m-1$), $C_m=\emptyset$,
$R_i=B_i\setminus C_i$ and $T_i=\bigsqcup_{j\preceq i}R_j$.
$T_i$ is defined recursively as
\[
 T_i=\begin{cases}
 R_i \sqcup \bigsqcup_{j\in k^{-1}(i)} T_j & (k^{-1}(i)\ne\emptyset), \\
 R_i                                       & (k^{-1}(i)=\emptyset).
 \end{cases}
\]
We state the proposed summation technique in general form in the theorem below.
\begin{theorem}
\label{thm:recursive}
Suppose that a sequence of sets $B_1,\ldots,B_m$ satisfies the running intersection property with respect to a function $k(\cdot)$.
Define $C_i$, $R_i$ and $T_i$ as before. 
Let $\chi_i(X_{B_i})=\chi_i(X_{C_i},X_{R_i})$ be an arbitrary function of $X_{B_i}=(X_{C_i},X_{R_i})$.
Suppose that
\[
 X_V \sim \mathrm{Mult}(N;p_V), \quad V=\bigcup_{i=1}^m B_i,
\]
is a multinomial random vector. 
Define
\begin{equation}
\label{xi0}
 \xi_i(N_i,X_{C_i}) = E\Bigl[{\prod}_{j\preceq i} \chi_j(X_{B_j}) \mid X_{C_i}, {\sum}_{j\in T_i} X_j = N_i \Bigr].
\end{equation}
In particular, define
\[
 \xi_m(N,\emptyset) = E\Bigl[{\prod}_{i=1}^m \chi_i(X_{B_i})\Bigr].
\]
Then
\begin{align}
& \xi_i(N_i,X_{C_i}) \nonumber \\
&= \begin{cases}
\displaystyle
 E^{\left(M_i,(N_j)_{j\in k^{-1}(i)}\right)|N_i}\Bigl[E^{X_{R_i}|M_i}\Bigl[\chi_i(X_{C_i},X_{R_i}){\prod}_{j\in k^{-1}(i)}\xi_j(N_j,X_{C_j})\Bigr]\Bigr] \\
\displaystyle
 & \hspace*{-20mm}
 (k^{-1}(i)\ne\emptyset), 
\\
\displaystyle
 E^{X_{R_i}|N_i} \bigl[\chi_i(X_{C_i},X_{R_i})\bigr]
 & \hspace*{-20mm}
 (k^{-1}(i)=\emptyset). \end{cases} \nonumber \\
&
\label{recursive}
\end{align}
Here, the conditional expectations in (\ref{recursive}) are with respect to
\begin{align}
(M_i,(N_j)_{j\in k^{-1}(i)})|N_i &\,\sim\, \mathrm{Mult}\biggr( N_i;\frac{(\sum_{j\in R_i}p_j,(\sum_{l\in T_j}p_l)_{j\in k^{-1}(i)})}{\sum_{j\in T_i} p_j} \biggl), \nonumber \\
X_{R_i}|M_i &\,\sim\,\mathrm{Mult}\biggl( M_i;\frac{p_{R_i}}{\sum_{j\in R_i} p_j} \biggr) \quad \mbox{for } k^{-1}(i)\ne\emptyset,
\label{condprob}
\end{align}
and
\begin{align*}
X_{R_i}|N_i &\,\sim\,\mathrm{Mult}\biggl( N_i;\frac{p_{R_i}}{\sum_{j\in R_i} p_j} \biggr) \quad \mbox{for } k^{-1}(i)=\emptyset.
\end{align*}
\end{theorem}

\begin{proof}
We prove the case $k^{-1}(i)\ne\emptyset$.
The case $k^{-1}(i)=\emptyset$ is similar and easier, and therefore omitted.

Noting that $\prod_{j\preceq i} \chi_j(X_{B_j}) = \chi_i(X_{B_i}) \prod_{j\in k^{-1}(i)} \prod_{l\preceq j} \chi_l(X_{B_l})$, we have
\begin{align}
 \xi_i(N_i,X_{C_i}) = E\Bigl[ & \chi_i(X_{B_i}) E\Bigl[ {\prod}_{j\in k^{-1}(i)} {\prod}_{l\preceq j} \chi_l(X_{B_l}) \nonumber \\
 & \mid X_{B_i},{\sum}_{l\in T_j} X_l = N_j\ (j\in k^{-1}(i)),X_{C_i},{\sum}_{j\in T_i} X_j = N_i \Bigr] \nonumber \\
 & \mid X_{C_i},{\sum}_{j\in T_i} X_j = N_i \Bigr].
\label{xi1}
\end{align}
As $B_i\supset C_i,R_i$ and $T_i = R_i \sqcup \bigsqcup_{j\in k^{-1}(i)} T_j$,
the conditions on $X_{B_i}$ and $\sum_{l\in T_j} X_l\,(=N_j)$ implies the conditions on $X_{C_i}$ and $\sum_{j\in T_i} X_j\,(=N_i)$.
That is,
\[
 N_i = M_i + {\sum}_{j\in k^{-1}(i)} N_j, \quad M_i={\sum}_{j\in R_i} X_j.
\]
Therefore, the inner conditional expectation on the right-hand side of (\ref{xi1}) is rewritten as
\[
 E\Bigl[ {\prod}_{j\in k^{-1}(i)} {\prod}_{l\preceq j} \chi_l(X_{B_l})
 \mid X_{B_i},{\sum}_{l\in T_j} X_l = N_j\ (j\in k^{-1}(i)) \Bigr].
\]
Moreover, when ${\sum}_{l\in T_j} X_l\,(=N_j)$ ($j\in k^{-1}(i)$) are given,
$X_{T_j}$ ($j\in k^{-1}(i)$) are independent,
and hence the expression above is rewritten as
\begin{align*}
{\prod}_{j\in k^{-1}(i)} & E\Bigl[ {\prod}_{l\preceq j} \chi_l(X_{B_l})
 \mid X_{B_i},{\sum}_{l\in T_j} X_l = N_j \Bigr] \\
&= {\prod}_{j\in k^{-1}(i)} E\Bigl[ {\prod}_{l\preceq j} \chi_l(X_{B_l})
 \mid X_{C_j},{\sum}_{l\in T_j} X_l = N_j \Bigr] \\
&= {\prod}_{j\in k^{-1}(i)} \xi_i(N_j,X_{C_j}).
\end{align*}
In the above, the second equality follows because of
$(\bigcup_{l\preceq j} B_l) \cap B_i=C_j$ (which will be proved later) and hence
$\prod_{l\preceq j} \chi_l(X_{B_l})$ is a function on $X_{B_i}$ through $X_{C_j}$.
Now, we have the formula
\begin{align*}
 \xi_i(N_i,X_{C_i}) = E\Bigl[ \chi_i(X_{C_i},X_{R_i}) {\prod}_{j\in k^{-1}(i)} \xi_i(N_j,X_{C_j}) \mid X_{C_i},{\sum}_{j\in T_i} X_j = N_i \Bigr].
\end{align*}
Actually this is equivalent to (\ref{recursive}), because under the conditional distribution, $(M_i,(N_j)_{j\in k^{-1}(i)})$ ($M_i=\sum_{j\in R_i}X_j$) and 
$X_{R_i}$ have the distributions given in (\ref{condprob}).

Finally, we prove $(\bigcup_{l\preceq j} B_l) \cap B_i=C_j$ for $i=k(j)$.
Obviously
$(\bigcup_{l\preceq j} B_l) \cap B_i \supset B_j\cap B_i = B_j\cap B_{k(j)} = C_j$.
It suffices to prove
$B_l \cap B_i \subset C_j$ for $l\preceq j$.
Recall that $l\preceq j$ implies $l<k(l)<k(k(l))<\cdots<k(k(\cdots(k(l))\cdots))=j$.
Also $i=k(j)$ implies $j<i$.
Hence, $B_l \cap B_i\subset B_l\cap (\bigcup_{h>l}B_h)=B_{k(l)}$,
$B_l \cap B_i = (B_l \cap B_i)\cap B_i \subset B_{k(l)}\cap B_i=B_{k(k(l))}$.
By repeatedly applying this manipulation, we reach $B_l \cap B_i\subset B_j$ and
$B_l \cap B_i = (B_l \cap B_i) \cap B_i\subset B_j\cap B_i=C_j$ follows.
\end{proof}

Theorem \ref{thm:recursive} provides an algorithm to compute the exact expectation
$E\bigl[{\prod}_{i=1}^m \allowbreak \chi_i(X_{B_i})\bigr]$ with $X_V\sim\mathrm{Mult}(N;p_V)$ by updating the tables $\xi_i(N_i,X_{C_i})$, $i=1,\ldots,m$.
We can evaluate the number of required summations as before.
The result is summarized below without proof.
\begin{theorem}
\label{thm:no_sum}
The number of summations required in the algorithm of Theorem \ref{thm:recursive} is
\begin{equation}
 \sum_{i=1}^{m-1} {N+|B_i|+|k^{-1}(i)| \choose |B_i|+|k^{-1}(i)|} +
 {N+|B_m|+|k^{-1}(m)|-1 \choose |B_m|+|k^{-1}(m)|-1} = O(N^{\mathrm{deg}})
\label{order}
\end{equation}
where
\begin{equation}
\mathrm{deg} = \max\left(\max_{1\le i\le m-1}(|B_i|+|k^{-1}(i)|),|B_m|+|k^{-1}(m)|-1\right).
\label{deg}
\end{equation}
\end{theorem}
Note that when $i<m$, the value $\xi_i(N_i,X_{C_i})$ is needed for $0\le N_i\le N$.
Whereas, when $i=m$, only the value for $N_i=N$ is needed.
This is the reason why the case $i=m$ is exceptional in (\ref{order}).

The number of summations (\ref{order}) is smaller than in the absent of this recursive computation technique:
\[
 {N+|V|-1 \choose |V|-1} = O(N^{|V|-1}), \quad V=\bigcup_{i=1}^m B_i.
\]

As shown, the proposed algorithm has an advantage in time complexity.
Whereas, it requires memory space to restore the tables.
Since in (\ref{xi0}), $N_i$ and the elements of $X_{C_i}=(X_i)_{i\in C_i}$ are arbitrary nonnegative integers such that $N_i+\sum_{i\in C_i}X_i\le N$, the size of the table for $\xi_i(N_i,X_{C_i})$ is
\[
 {N+|C_i|+1 \choose |C_i|+1} = O(N^{|C_i|+1}).
\]
However, in the process of computation of (\ref{recursive}) for $i=1,\ldots,m$, the table $\xi_i(N_i,X_{C_i})$ can be deleted (i.e., the memory space can be released) once $\xi_{k(i)}(N_{k(i)},X_{C_{k(i)}})$ is computed.
Therefore, the problem of space complexity does not matter in practice.

\subsection{A class of distributions for recursive computation}

Although Theorem \ref{thm:recursive} is stated for the multinomial distribution,
it also works for a class of distributions including the multinomial distribution.

Let $B_1$ and $B_2$ be index sets such that $V=B_1\cup B_2$,
and let $C_1 = B_1 \cap B_2$, $R_1 = B_1 \setminus C_1$ and $R_2=B_2$ again.
Suppose that $X_i$ is distributed independently for $i\in V$ according to a certain distribution.
Under the conditional distribution where $N=\sum_{i\in V}X_i$ is given,
if we pose additional conditions that $(\sum_{i\in R_2}X_i,\sum_{i\in R_1}X_i)=(M_2,M_1)$ is given, then $X_{R_1}$ and $X_{R_2}$ become independent.
Therefore, the three steps (\ref{M2M1})--(\ref{XR1}) for generating random numbers,
and the corresponding decomposition of the expectation continue to hold for an arbitrary distribution of $X_i$.
If explicit expressions for probability density functions of the conditional distributions
$X_V|N$, $X_{R_2}|M_2$, $X_{R_1}|M_1$, and $(M_2,M_1)|N$ are available,
we have the the computation formula of the type (\ref{E}).

In general, if some explicit formula is available for the probability density function of the conditional distribution
\begin{equation}
\label{generalcond}
 \Bigl({\sum}_{j\in R_i}{X_j}\Bigr)_{i\ge 1}\,\Big|_{N},
\end{equation}
where the $R_i$ are subsets of $V$ such that $V=\bigsqcup_{i\ge 1}R_i$,
we can construct the recurrence computation formula of the type of Theorem \ref{thm:recursive}.

The class of distributions having the explicit conditional density function of (\ref{generalcond}) includes
the normal distribution, the Gamma distribution, the binomial distribution and the negative binomial distribution. 
The conditional distributions of (\ref{generalcond}) corresponding to the above four distributions are the (degenerate) normal distribution, the Dirichlet distribution, the multivariate hypergeometric distribution, and the Dirichlet-multinomial distribution, respectively.
For these distributions, the recurrence computation formula of Theorem \ref{thm:recursive} works by replacing summations with integrations when the distribution is continuous.

\section{Extraction of Markov structure}
\label{sec:markov}

As shown in the previous section, when the sequence $Z_1,\ldots,Z_M$ of subsets of $V$ has the running intersection property, the expectation (\ref{expectation}) can be evaluated efficiently with the recursive summation/integration technique proved in Theorem \ref{thm:recursive}.
However, in general, $Z_1,\ldots,Z_M$ does not have this property.
To apply this technique to general cases, one method is to prepare another sequence $B_1,\ldots,B_m$ having the running intersection property such that each $Z_j$ is included into at least one of $B_1,\ldots,B_m$.
If we have such $B_i$, by defining a function
$\tau:\{1,\ldots,M\}\to\{1,\ldots,m\}$ so that $Z_j \subset B_{\tau(j)}$,
the expectation (\ref{expectation}) is written as
\[
 E\Bigl[{\prod}_{i=1}^m \chi_i(X_{B_i})\Bigr], \quad
 \chi_i(X_{B_i}) := {\prod}_{j\in \tau^{-1}(i)}\chi_j(X_{Z_j}),
\]
which can be dealt with by Theorem \ref{thm:recursive}.

Such a sequence $B_i$, $i=1,\ldots,m$, is known to be obtained through a chordal extension in Algorithm \ref{alg:chordal_extension}.
Here, we summarize some notions and basic facts on chordal graphs.

Let $G = (V,E)$ be a connected undirected graph, where $V$ is a set of
vertices and $E \subset V \times V$ is a set of edges.
$G$ is called complete if every pair of vertices is joined by an edge. 
For a subset $B \subseteq V$, let $G(B)$ denote a subgraph induced by $B$. 
That is,
\[
 G(B) = (B,E(B)),\quad E(B) = \{(i,j)\in E \mid i,j\in B\}.
\]
When $G(B)$ is complete, $B$ is called a clique. 
A clique that is maximal with respect to inclusion is called a maximal clique.
$G$ is called chordal if every cycle with length greater than three has a chord.  

Let $B_1,\ldots,B_m$ be a sequence of all maximal cliques of $G$ satisfying the running intersection property (\ref{rip}) in Definition \ref{def:rip}.
The sequence is called perfect if for every $i=1,\ldots,m$, $B_i \cap B_{k(i)}$ is a clique of $G$.
Denote by ${\cal B}=\{B_1,\ldots,B_m\}$ the set of maximal cliques of $G$.
A vertex $v \in V$ is called simplicial if its adjacent vertices form a clique in $G$.
A perfect elimination ordering $v_1,\ldots,v_n$ of $G$ is an ordering of vertices of $G$ such that for every $i$, $v_i$ is a simplicial in $G(\{v_{i+1},\ldots,v_n\})$. 

Let $\mathrm{Simp}(B_i)$ denote the set of simplicial vertices of $G$ in $B_i$.
Then $B_i$ is called a boundary clique if there exists $B_j$, $j\ne i$, such that 
\[
B_i \cap (V \setminus \mathrm{Simp}(B_i)) = B_i \cup B_j.
\]

The following proposition on the property of chordal graphs is well known
 (e.g., \cite{Lauritzen96} and \cite{Buneman74}). 
\begin{proposition}
Let $G$ be an undirected graph.
The four statements below are equivalent.

(i) $G$ is chordal.

(ii) $G$ has a perfect sequence of the maximal cliques.

(iii) $G$ vertices can be ordered to have a perfect elimination ordering.
\end{proposition}

We generate a perfect sequence $B_1,\ldots,B_m$ from $Z_1,\ldots,Z_M$ according to the following procedure.
\begin{algorithm}\ %
\label{alg:chordal_extension}
\begin{itemize}
\item[Step 0.]
Define an undirected graph $G=(V,E)$ with vertices $V=\{1,\ldots,n\}$ and edges
\[
 E = \{ (i,j)\in V\times V \mid i,j\in Z \ \mbox{for some}\ Z\in\mathcal{Z} \},
\]
where $\mathcal{Z}=\{ Z_1,\ldots,Z_M \}$.

\item[Step 1.]
Add edges of $E_1$ to $G$ so that the extended graph $\widetilde G=(V,E\cup E_1)$ is a chordal graph.

\item[Step 2.]
Identify the perfect sequence of the maximal cliques $B_1,\ldots,B_m$ of $\widetilde G$.
This sequence has the running intersection property and $V=\bigcup_{i=1}^m B_i$.
\end{itemize}
\end{algorithm}

The procedure for Step 1 is referred to as the \textit{chordal extension\/}.
Constructing the chordal extension such that the maximum size of the maximal clique is minimum is known to be a NP-hard problem \citep{Fukuda-etal00}.
In this section, we propose a heuristic method to construct an approximately best chordal extension.
For Step 2, we provide a method proved in Theorem \ref{thm:hara}.
The other method for the same purpose based on the maximum cardinality search procedure is also known \citep{Blair-Peyton93}.

We now explain Steps 0--2 in detail using a small example.

\subsubsection*{Step 0: Defining an undirected graph}

Suppose that the vertex set is $V=\{1,2,\ldots,9\,(=n)\}$, and the family of subsets of $V$ is given by
\begin{align}
\mathcal{Z} = \Bigl\{ &
\{1\}, \{2\}, \{3\}, \{4\}, \{5\}, \{6\}, \{7\}, \{8\}, \{9\},
 \nonumber \\
& \{4,5\}, \{7,8\}, \{4,8\}, \{3,7\}, \{4,5,8\}, \{2,4\}, \nonumber \\
& \{1,3\}, \{2,3\}, \{2,4,5\}, \{3,6\}, \{8,9\} \Bigr\}.
\label{Z20}
\end{align}
The associated undirected graph $G$ is shown in Figure \ref{fig:G} (left).

\subsubsection*{Step 1-a: Renumbering of vertices}

Step 1 is divided into three parts.
In this substep, we renumber the vertices.
Consider the following procedure to make all vertices removed from the graph $G$ sequentially:
First, find a vertex, say $v_1$, that has the minimum number of adjacent edges in the graph $G$.
Then, remove $v_1$ as well as its adjacent edges $E_1 = \{ (v_1,i) \in E \}$ from the graph $G$.
Then, find a vertex, say $v_2$, that has the minimal number of adjacent edges in the graph $(V\setminus\{v_1\},E\setminus E_1)$ with $n-1$ vertices.
Continuing with this procedure, we have a sequence of vertices $v_1,v_2,\ldots,v_n$.
This renumbering is called the minimum degree (MD) ordering \citep{Parter61}.

The MATLAB function {\tt symamd} is available to obtain an MD ordering \citep{Redfern-Campbell98}.
In our example,
\[
\begin{pmatrix} i \\ v_i \end{pmatrix} =
\begin{pmatrix}
 1 & 2 & 3 & 4 & 5 & 6 & 7 & 8 & 9 \\
 9 & 1 & 6 & 3 & 7 & 2 & 4 & 5 & 8
\end{pmatrix}.
\]
We illustrate the resulting renumbered graph in Figure \ref{fig:G} (right).

\begin{figure}[H]
\begin{center}
\begin{tabular}{lll}
 \includegraphics[width=0.4\linewidth]{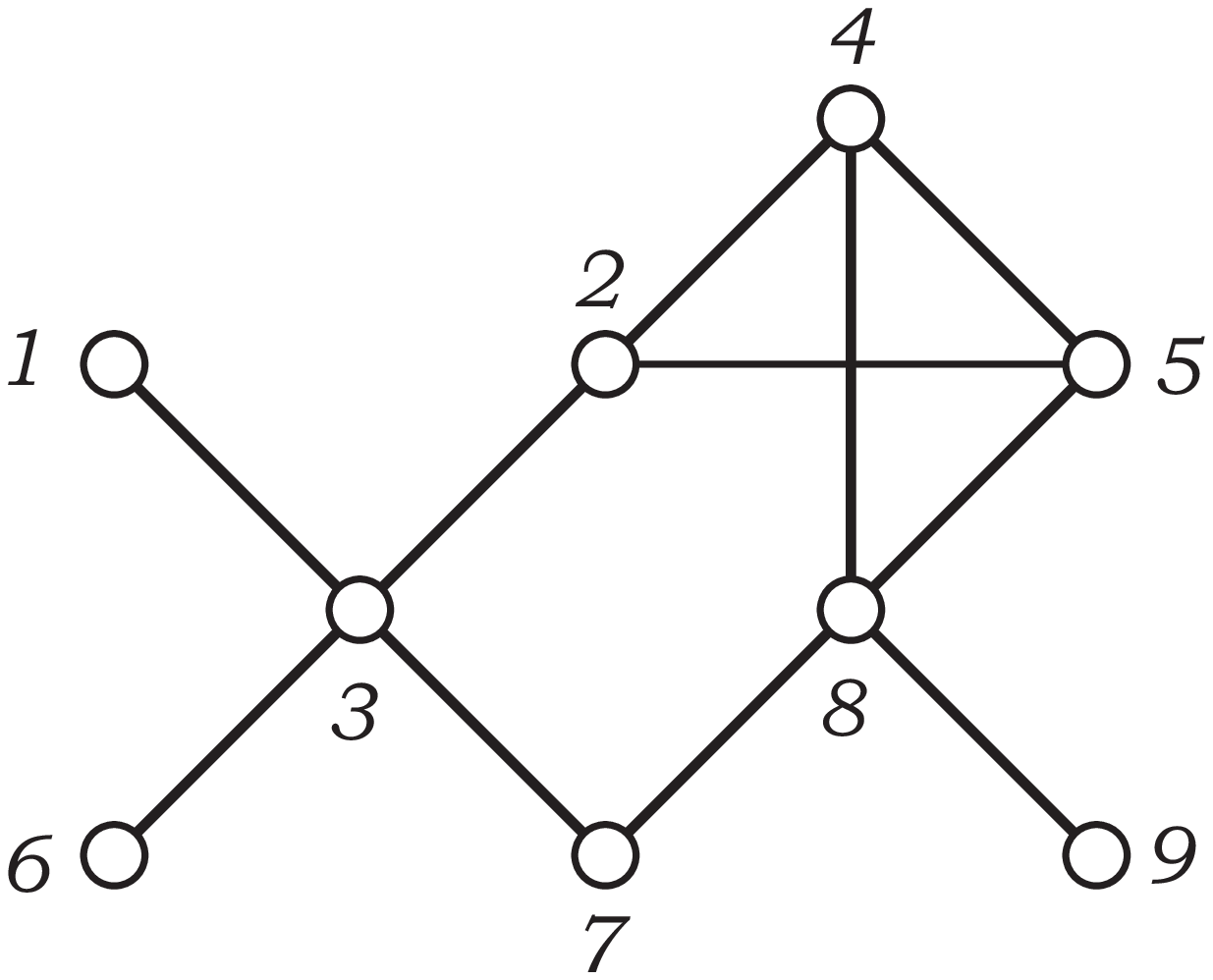} & &
 \includegraphics[width=0.4\linewidth]{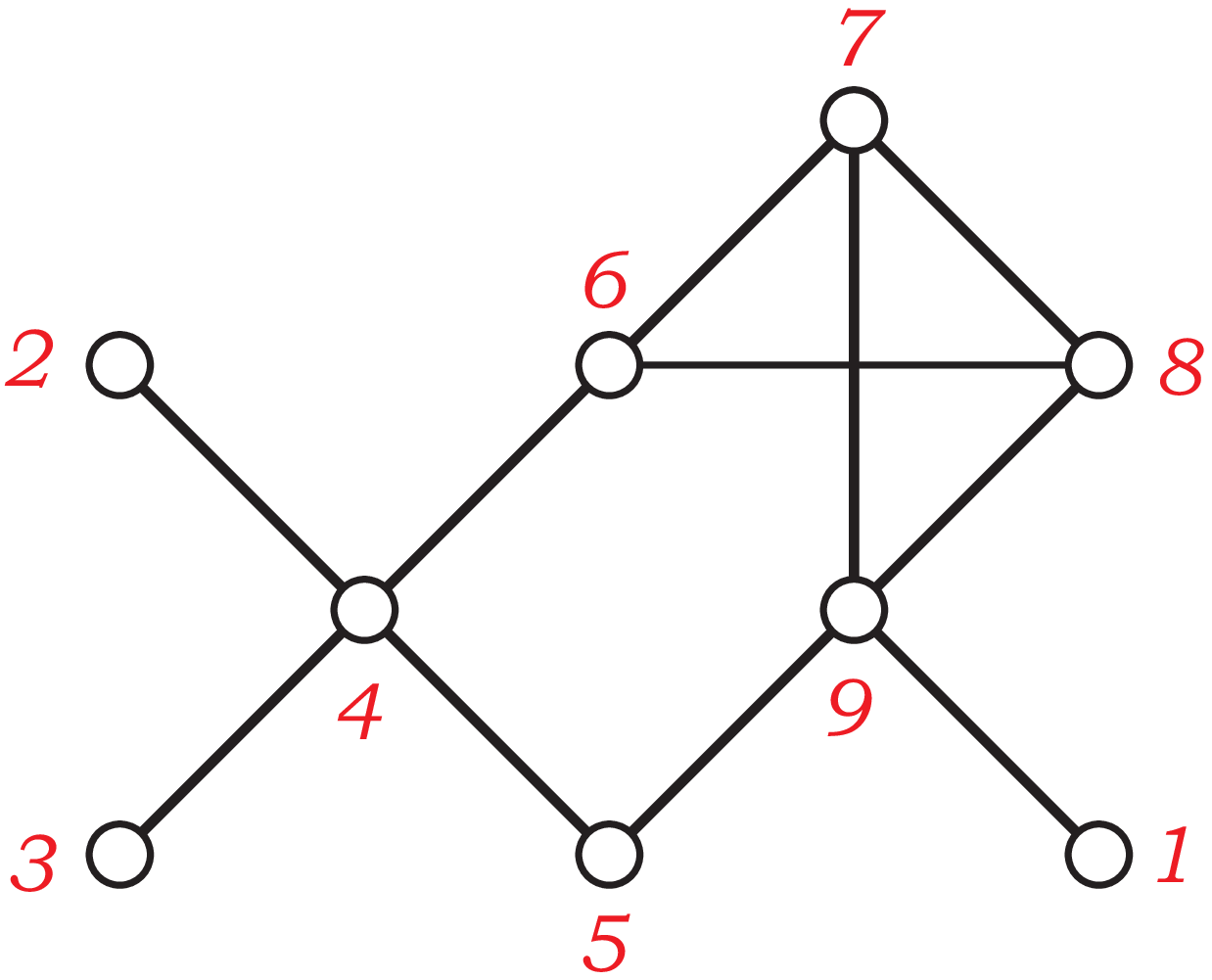} \\
\end{tabular}
\end{center}
\vspace{-0.3cm}
\caption{Graph $G$
 (left: original, right: after renumbering vertices).}
\label{fig:G}
\end{figure}

\subsubsection*{Step 1-b: Triangularization}

In this substep, we add additional edges to the graph $G$. 
For $i=1,\ldots,n$, if there exist $j,k$ such that $i<j<k$,
 $(i,j)\in E$ and $(i,k)\in E$, then add the edge $(j,k)$ to $G$ (if it does not exist).
This procedure is referred to as \textit{triangularization\/}.
Let $\widetilde G=(V,\widetilde E)$ be the resulting graph after the triangularization process.
From this construction, for each $i$, the set of vertices
\begin{equation}
\label{simplicial}
 V_i = \Bigl\{ j\in V \mid j>i,\ (i,j)\in \widetilde E \Bigr\}\cup\{i\}
\end{equation}
forms a clique.
That is, the vertices are ordered to be a perfect elimination ordering.
Figure \ref{fig:G-chordal} (left) depicts the triangulated graph $\widetilde G$ with added edges in red.

In our example, two edges $(5,6)$ and $(6,9)$ are added.
$V_i$ in (\ref{simplicial}) are
\begin{align*}
V_1 &=\{1,9\},\ V_2=\{2,4\},\ V_3=\{3,4\},\ V_4 =\{4,5,6\},\\
V_5 &=\{5,6,9\},\ V_6=\{6,7,8,9\},\ V_7=\{7,8,9\},\ V_8=\{8,9\},\ V_9=\{9\}.
\end{align*}

\subsubsection*{Step 1-c: Finding maximal cliques}

Let $V_1,\ldots,V_n$ be the sequence of cliques defined in (\ref{simplicial}).
From these $n$ cliques, we remove all ``non-maximal'' $V_i$ such that 
\[
 V_i \subset V_l \quad\text{for some $l<i$}.
\]

In our example, the maximal cliques are
\begin{align*}
V_1 &=\{1,9\},\ V_2=\{2,4\},\ V_3=\{3,4\},\ V_4 =\{4,5,6\},\\
V_5 &=\{5,6,9\},\ V_6=\{6,7,8,9\}.
\end{align*}

\subsubsection*{Step 2: Identifying a perfect sequence of the maximal cliques}

Suppose that there are $m$ ($\le n$) remaining cliques after removing the non-maximal $V_i$.
Among the set of maximal cliques, we introduce an order defined below.
For two cliques $V=\{v_1,\ldots,v_l\}$ and $V'=\{v'_1,\ldots,v'_{l'}\}$
such that $v_1<\cdots<v_l$ and $v'_1<\cdots<v'_{l'}$, define a lexicographic order:
\begin{equation}
\begin{split}
V<V' \quad \mbox{iff}\quad
               & v_l<v_{l'} \\
&\text{or}\ \  v_l=v_{l'},\ v_{l-1}<v_{l'-1} \\
&\text{or}\ \  v_l=v_{l'},\ v_{l-1}=v_{l'-1}, v_{l-2}<v_{l'-2} \\
&\text{or}\ \  \cdots.
\end{split}
\label{lex}
\end{equation}
This comparison procedure should stop properly if one of two is not a proper subset of the other.
According to this order, we order the $m$ maximal cliques $B_1,\ldots,B_m$, which is shown to have the running intersection property by Theorem \ref{thm:hara} below.

The function $k(\cdot)$ is defined as
\begin{equation}
 \label{k(i)}
 k(i) = \min\Bigl\{ k>i \mid B_i\cap\Bigl({\bigcup}_{j>i} B_j\Bigr) \subset B_k \Bigr\}.
\end{equation}

In our example, we obtain the final results
\begin{align}
B_1 &= \{2,4\},\ B_2=\{3,4\},\ B_3=\{4,5,6\},\ B_4 =\{1,9\}, \nonumber \\
B_5 &= \{5,6,9\},\ B_6=\{6,7,8,9\},
\label{Bi}
\end{align}
\[
\begin{pmatrix} i \\ k(i) \end{pmatrix} =
\begin{pmatrix}
 1 & 2 & 3 & 4 & 5 & 6 \\
 2 & 3 & 5 & 5 & 6 & -
\end{pmatrix},
\]
or
\begin{align}
 & k^{-1}(1)=\emptyset,\ k^{-1}(2)=\{1\},\ k^{-1}(3)=\{2\},\ k^{-1}(4)=\emptyset, \nonumber \\
 & k^{-1}(5)=\{3,4\},\ k^{-1}(6)=\{5\}
\label{ki}
\end{align}
(see Figure \ref{fig:G-chordal} (right)).

\begin{figure}[H]
\begin{center}
\begin{tabular}{lll}
 \includegraphics[width=0.4\linewidth]{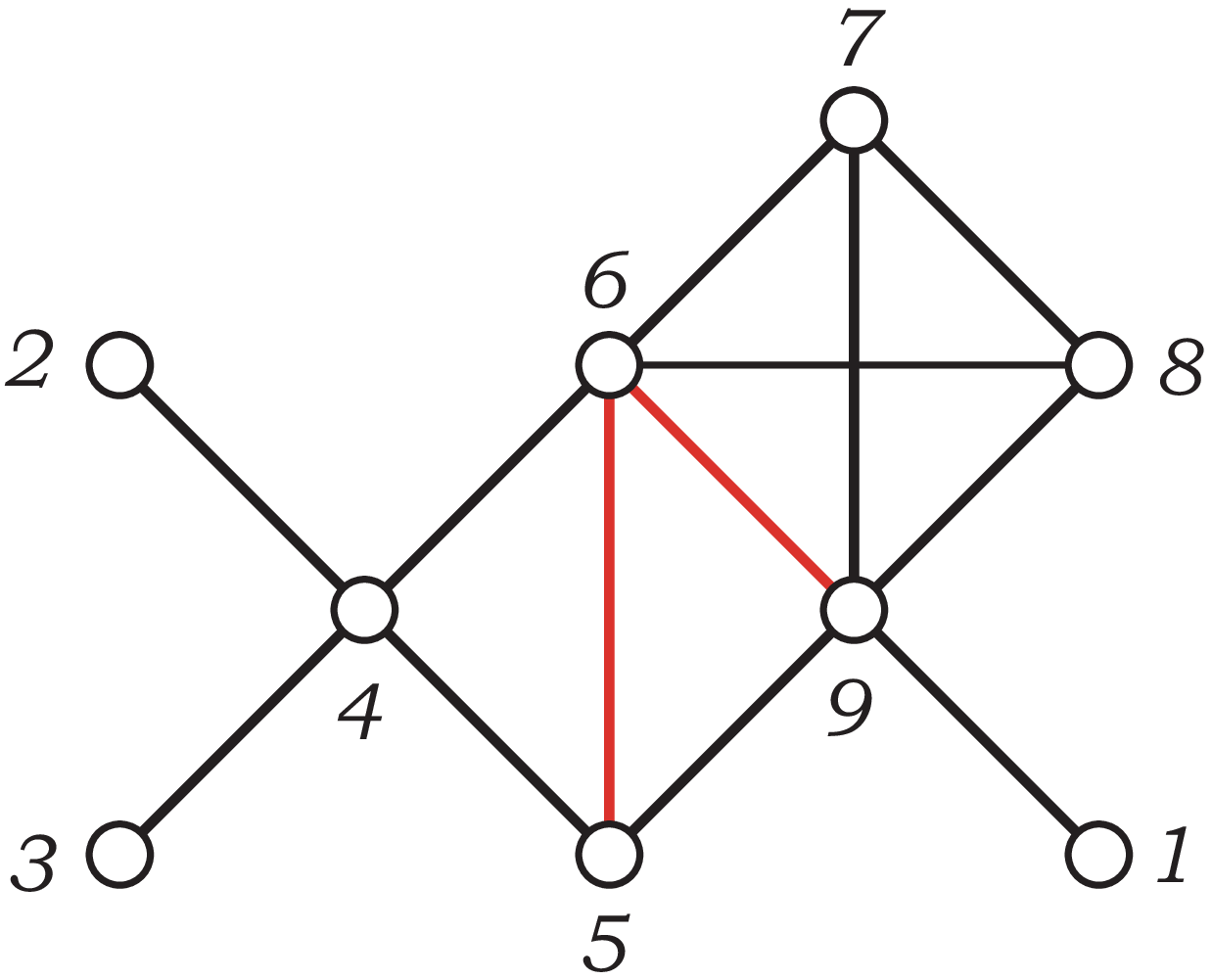} & &
 \includegraphics[width=0.4\linewidth]{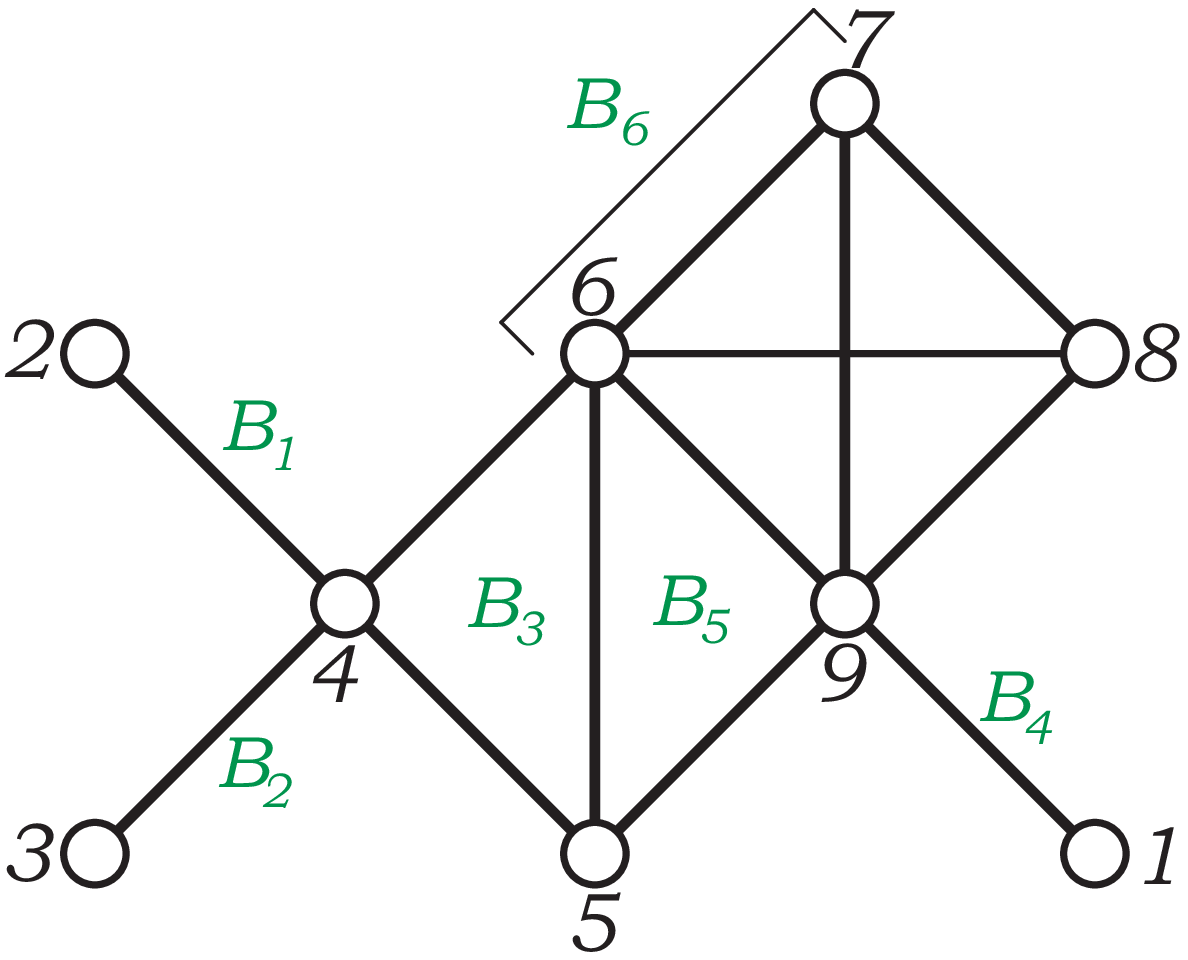} \\
\end{tabular}
\end{center}
\vspace{-0.3cm}
\caption{Chordal graph $\widetilde G$
 (left: added edges are in red, right: maximal cliques are in green).}
\label{fig:G-chordal}
\end{figure}

The required number of summations is estimated by Theorem \ref{thm:no_sum}.
In our example, $B_i$ and $k^{-1}(i)$ are given in (\ref{Bi}) and (\ref{ki}).
For example, when $N=\sum_i X_i = 28$, there are 314,621 summations using the proposed method, and 30,260,340 summations when we do not.

The theorem below validates Step 2 of the proposed algorithm.

\begin{theorem}
\label{thm:hara}
The sequence of the maximal cliques $B_1,\ldots,B_m$ identified in Step 2 has the running intersection property.
That is, if the partial order (\ref{lex}) can be embedded into the total order $B_1,\ldots,B_m$, then
$B_i \cap \bigl(\bigcup_{j>i} B_j\bigr) \subset B_{k(i)}$ for $k(i)>i$.
\end{theorem}

\begin{proof}
We prove this proposition by induction. The case where $m=1$ is trivial. 
Suppose the proposition holds for all chordal graphs with up to $m-1$ maximal cliques.
 
Assume that there does not exist $k(1)$ in (\ref{k(i)}). 
Then there exists more than one clique $H_1,\ldots,H_t$, $t>1$, such that
\[
 C_1 := B_1 \cap (B_2 \cup \cdots \cup B_m) 
 = B_1 \cap (H_1 \cup \cdots \cup H_t), 
\]
\[
 B_1 \cap H_i \ne B_1 \cap H_j, \quad i \ne j
\]
and 
\[
 B_1 \cap H_i \subset C_1, \quad i = 1,\ldots, t. 
\]
Since every vertex of $C_1$ is included in more than one clique, it is not simplicial in $\widetilde G$ (e.g., \cite{Hara-Takemura06}).
In other words, in the subgraph $\widetilde G(B_1 \cup H_1 \cup \cdots \cup H_t)$,
every element of $C_1$ is not simplicial. 
This shows that there exists a $H_i$ such that for all $v \in H_i \setminus C_1$ and for all $v' \in C_1$, $v<v'$. 
However, this implies $B_1 > H_i$ in the sense of (\ref{lex}), which is a contradiction.
Therefore, there exists $k(1)$ in (\ref{k(i)}). 
 
This also shows that $B_1$ is a boundary clique.
We then have
\[
 \widetilde G(V \setminus \mathrm{Simp}(B_1)) = 
 \widetilde G(B_2 \cup \cdots \cup B_m) 
\]
and $\widetilde G(B_2 \cup \cdots \cup B_m)$ is a connected chordal graph with the set of maximal cliques $\{B_2,\ldots,B_m\}$ and $B_2 < \cdots < B_m$
 (e.g., \cite{Hara-Takemura06}). 
From the inductive assumption, a sequence s$B_2,\ldots,B_m$ is perfect. 
Therefore, $B_1,B_2,\ldots,B_m$ is also perfect.
\end{proof}

\section{Illustrative data analysis with real data}
\label{sec:examples}

\subsection{\cite{Kulldorff97}'s scan statistic}
\label{subsec:scan}

In this section, we provide illustrative data analyses with real data.
As stated in Section \ref{sec:introduction}, in spatial epidemiology,
\cite{Kulldorff97}'s statistic serves as a standard scan statistic.
Using the notations $X_V=(X_i)_{i\in V}$, $p_V=(p_i)_{i\in V}$, and $Z\subset V$ defined in Section \ref{sec:introduction},
this statistic is written as
\begin{equation}
\varphi_N(X_Z) = \begin{cases}
N \Bigl\{ p\left( \frac{\widehat p}{p}\log \frac{\widehat p}{p} -
     \frac{\widehat p}{p} + 1 \right) + (1-p) & \\
 \quad \times \left( \frac{1-\widehat p}{1-p}\log \frac{1-\widehat p}{1-p} -
                     \frac{1-\widehat p}{1-p} + 1 \right) \Bigr\},
   & \hspace*{-5mm} \mbox{if } \frac{\widehat p}{p}\ge 1, \\ 
 0,& \hspace*{-5mm} \mbox{otherwise}, \end{cases}
\label{Kulldorff}
\end{equation}
where
$\widehat p=\sum_{i\in Z} X_i/\sum_{i\in V} X_i$ and $p = \sum_{i\in Z}\lambda_i/\sum_{i\in V}\lambda_i$.
We employ this statistic in this section.

\subsection{Implementation}
\label{subsec:implementation}

The recursive summation algorithm for Theorem \ref{thm:recursive} is implemented in C.
The algorithm from Section \ref{sec:markov} to construct a sequence $B_1,\ldots,B_m$ having the running intersection property is implemented in MATLAB.

For $\mathcal{Z}$ given in (\ref{Z20}), if we suppose
\[
 \textstyle
 (X_i) = (2, 7, 7, 2, 2, 2, 2, 2, 2), \quad N=\sum_i X_i = 28,
\]
and $\lambda_i \equiv 1$, the maximum of the scan statistic is
$\max_{Z\in\mathcal{Z}}\varphi_Z(X_Z) = 5.167364$,
and the maximum is attained at $Z=\{2,3\}$.
The $p$-value is 0.01371293.
The number of summations is 314,621 (coinciding with the value by (\ref{order})).
The computational time with C is less than 1 sec (MacBook Air 11-inch Early 2014 1.4GHz Intel Core i5). 
Without using the proposed algorithm, the number of summations is 30,260,340,
and the computational time is 130 sec (ibid). 

\subsection{Monthly frequencies of spontaneous abortions}
\label{subsec:spontaneous_abortion}

The method of spatial scan is applicable to detect clustering in the time domain.
Figure \ref{fig:page84-barplot} depicts monthly frequencies of trisomy among karyotyped spontaneous abortions of pregnancies.
The data are taken over 24 months from July 1975 to June 1977 in three New York hospitals.
Total count is $N=62$.
The data are tabulated in \cite{Tango84}.
Using these, we detect the clustering in the time domain,
that is, the period with a frequency higher than normal.

\begin{figure}[H]
\begin{center}
\includegraphics[scale=0.578125]{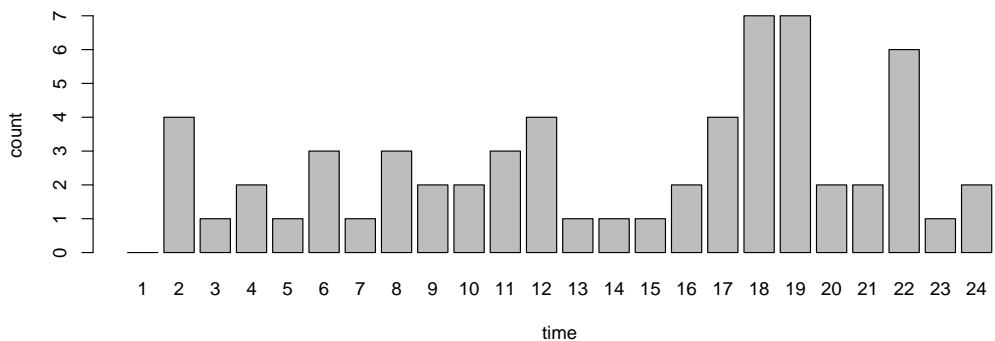} 
\end{center}
\caption{Monthly frequencies of trisomy among karyotyped spontaneous abortions of pregnancies.}
\label{fig:page84-barplot}
\end{figure}

\cite{Wallenstein80} pointed out that the maximum frequency per consecutive two months is 14 for the time $t=18,19$.
He considered the conditional distribution of the maximum number of events for any arbitrary two-month period with $N=62$ given under the condition that events occur independently and uniformly over every 24 months, and reported that a value of 14 corresponds the $p$-value of $0.038$.
For the $p$-value formula, see \cite{Wallenstein-Naus73} and \cite{Neff-Naus80}.

For the same data, we apply our method to obtain the exact $p$-value.
Let $L$ be the maximum window size, and we consider frequencies during the period less than or equal to $L$.
For example, when $L=3$, the scan window consists of all of one point of time, successive two-, and three-month periods.
That is,
\[
 \mathcal{Z} = \bigl\{ \{1\},\ldots,\{24\},\{1,2\},\ldots,\{23,24\},\{1,2,3\},\ldots,\{22,23,24\}\bigr\}.
\]
The chordal graph used here for the calculation is given in Figure \ref{fig:page84-chordal}.

\begin{figure}[H]
\begin{center}
\includegraphics[scale=0.25]{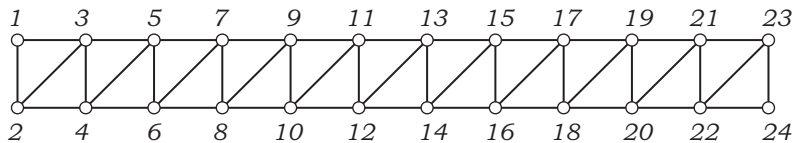}
\end{center}
\caption{Chordal graph for the temporal clustering.}
\label{fig:page84-chordal}
\end{figure}

We apply \cite{Kulldorff97}'s scan statistic in (\ref{Kulldorff}) with the assumption that $\lambda_i$ are constant.
The largest five statistics for $L=5$ are listed in Table \ref{tab:largest5}.
The last four columns provide the corresponding $p$-values,
which are also evaluated under the cases where $L$ is less than or equal to five.

\begin{table}[H]
\caption{Largest five statistics and their $p$-values.}
\label{tab:largest5}
\ \vspace*{-5mm}
\begin{center}
\begin{tabular}{cccccc}
\bhline{1pt}\rule{0pt}{2.2ex}
statistic & period & 
$L=5$ & $L=4$ & $L=3$ & $L=2$ \\
\hline
5.954    & 17,18,19       & 0.0175    & 0.0151    & 0.0135    & NA \\
%
5.847    & 18,19          & 0.0217    & 0.0194    & 0.0180    & 0.0140    \\
%
5.143    & 18,19,20,21,22 & 0.0453    & NA        & NA & NA \\
%
4.507    & 17,18,19,20    & 0.0716    & 0.0695    & NA & NA \\
%
4.507    & 16,17,18,19    & 0.0716    & 0.0695    & NA & NA \\
\bhline{1pt}
\end{tabular}
\end{center}
\end{table}

Table \ref{tab:largest5} suggests that the period $t=17,18,19$ is detected as a temporal cluster with the smallest $p$-value of 0.0135.

\subsection{Gallbladder cancer in Yamagata}
\label{subsec:yamagata}

Figure \ref{fig:yamagata-smr} depicts the choropleth map of standardized mortality ratios (SMRs) for gallbladder cancer (among male) analyzed in \cite{Tango-etal07} and \cite{Tango10}.
The data are for Yamagata Prefecture, Japan, which consists of 44 municipalities (villages, towns, and cities), for the period 1996--2000, based on the age-specific mortality rates from the 1985 national census population.
The total observed number of deaths was $N=147$.

\begin{figure}[H]
\begin{center}
\hspace*{-4cm}
\includegraphics[angle=-90,scale=0.8]{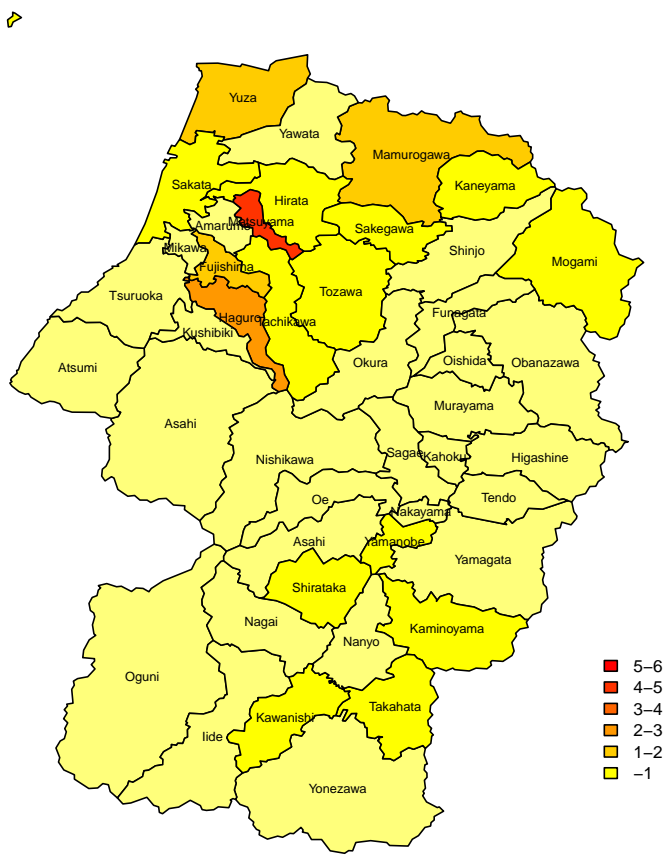}
\vspace*{-2cm}
\end{center}
\caption{SMRs of gallbladder cancer (among male) in Yamagata Prefecture 1996--2000.}
\label{fig:yamagata-smr}
\end{figure}

Figure \ref{fig:yamagata-lrt} depicts the scan statistics $\varphi_N(X_Z)$ when $Z$ consists of one district under the same condition as Figure \ref{fig:yamagata-smr}.

\begin{figure}[H]
\begin{center}
\hspace*{-4cm}
\includegraphics[angle=-90,scale=0.8]{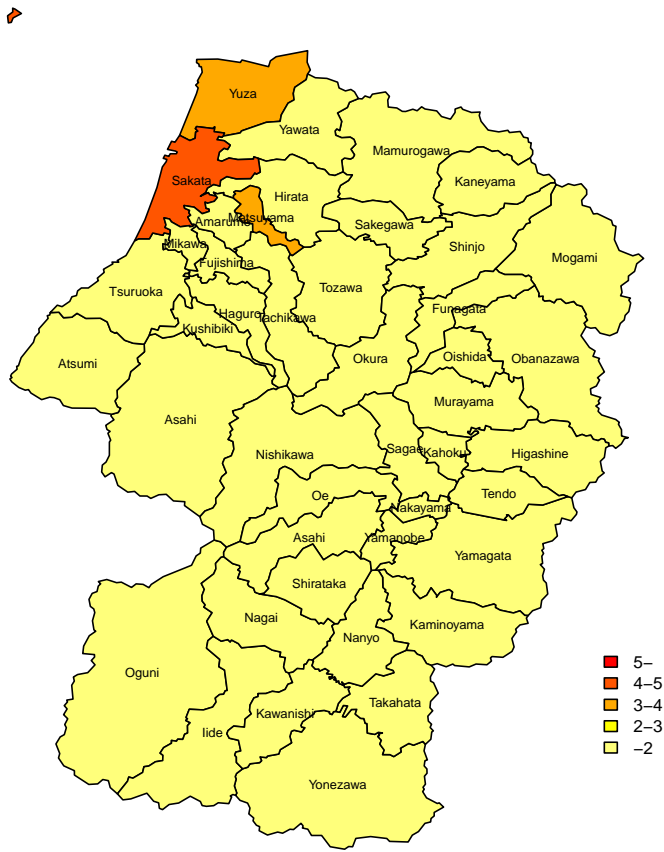}
\vspace*{-2cm}
\end{center}
\caption{Scan statistic $\varphi_N(X_Z)$ when $Z$ consists of one district.}
\label{fig:yamagata-lrt}
\end{figure}

For these data, we apply two kinds of scan windows.
(i) Each window $Z$ consists of one district (i.e., $|Z|=1$),
and hence the number of windows is equal to the number of districts.
(ii) Each window consists of one district, or two districts adjacent to each other (i.e., $|Z|\le 2$).

The number of scan windows is (i) 44 ($|Z|=1$) and (ii) 154 ($|Z|\le 2$).
Unfortunately, the computation is infeasible for the latter case because the number of summations is more than $10^{14}$ (see Table \ref{tab:graph-summary}).
Therefore, we randomly divided the 154 scan windows into two groups (groups 1 and 2), and calculated the $p$-values for each group and summed them.
This yields conservative $p$-values because this manipulation is nothing more than the Bonferroni correction.
The generated graphs and their chordal extensions are shown in Figure \ref{fig:part}, and their features are summarized in Table \ref{tab:graph-summary}.

\begin{table}[H]
\caption{Number of scan windows and generated graphs.}
\label{tab:graph-summary}
\ \vspace*{-5mm}
\begin{center}
\begin{small}
\begin{tabular}{l|cccccccr}
\bhline{1pt}\rule{0pt}{2.2ex}
 & $N$ & $M$ & $n$ & $e$ & $\widetilde e-e$ & $m$ & deg & \# of summations \\
\hline
whole data & 147 & 154 & 44 & 110 & 47 & 35 & 8 & 141,445,034,516,085 \\
group 1    & 147 &  76 & 44 &  56 & 16 & 33 & 6 & 82,837,604,771 \\
group 2    & 147 &  78 & 44 &  48 & 10 & 33 & 6 & 35,091,700,432 \\
\bhline{1pt}
\end{tabular}

\medskip
\parbox[t]{12cm}{\small
$N$: \# of total events ($=\sum X_i$),
$M$: \# of scan windows, $n$: \# of vertices, \\ $e$: \# of edges of original graph $G$, $\widetilde e$: \# of edges in the chordal graph $\widetilde G$, \\
$m$: \# of cliques of $\widetilde G$, deg: defined in (\ref{deg}).}
\end{small}
\end{center}
\end{table}

Table \ref{tab:largest} shows several of the largest scan statistics for the windows $|Z|=1$ and $|Z|\le 2$, and the corresponding $p$-values.
The district of Sakata and Yuza is detected as a spatial disease cluster with the smallest $p$-value of 0.00953.


\begin{table}[H]
\caption{Several of the largest statistics and their $p$-values.}
\label{tab:largest}
\ \vspace*{-5mm}
\begin{center}
\begin{small}
\begin{tabular}{cccc}
\bhline{1pt}\rule{0pt}{2.2ex}
statistic & districts & $|Z|\le 2$ & $|Z|=1$ \\
\hline
7.651    & \{Sa,\,Yu\} & \ 0.00953    & NA \\ 
%
4.578    & \{Sa\}      & 0.1847     & 0.0433    \\ 
%
4.356    & \{Sa,\,Hi\} & 0.2247     & NA       \\ 
%
4.247    & \{Sa,\,Mi\} & 0.2541     & NA       \\ 
%
3.924    & \{Sa,\,Am\} & 0.3444     & NA       \\ 
%
3.570    & \{Sa,\,Ya\} & 0.4795     & NA       \\ 
%
3.364    & \{Ma\}      & 0.5699     & 0.1739   \\ 
%
3.205    & \{Fu,\,Ha\} & 0.6458     & NA       \\ 
%
3.071    & \{Yu\}      & 0.7025     & 0.2065   \\ 
\bhline{1pt}
\end{tabular}

\medskip
\parbox[t]{11cm}{\small
Sa: Sakata (municipality code: 06204), Am: Amarume (06422), \\ Fu: Fujishima (06423), Ha: Haguro (06424), Mi: Mikawa (06426), \\ Yu: Yuza (06461), Ya: Yawata (06462), Ma: Matsuyama (06463), \\ Hi: Hirata (06464).}
\end{small}
\end{center}
\end{table}

\begin{figure}[H]
\begin{center}
\begin{tabular}{ll}
 \hspace*{-35mm}
 \includegraphics[angle=-90,scale=0.5]{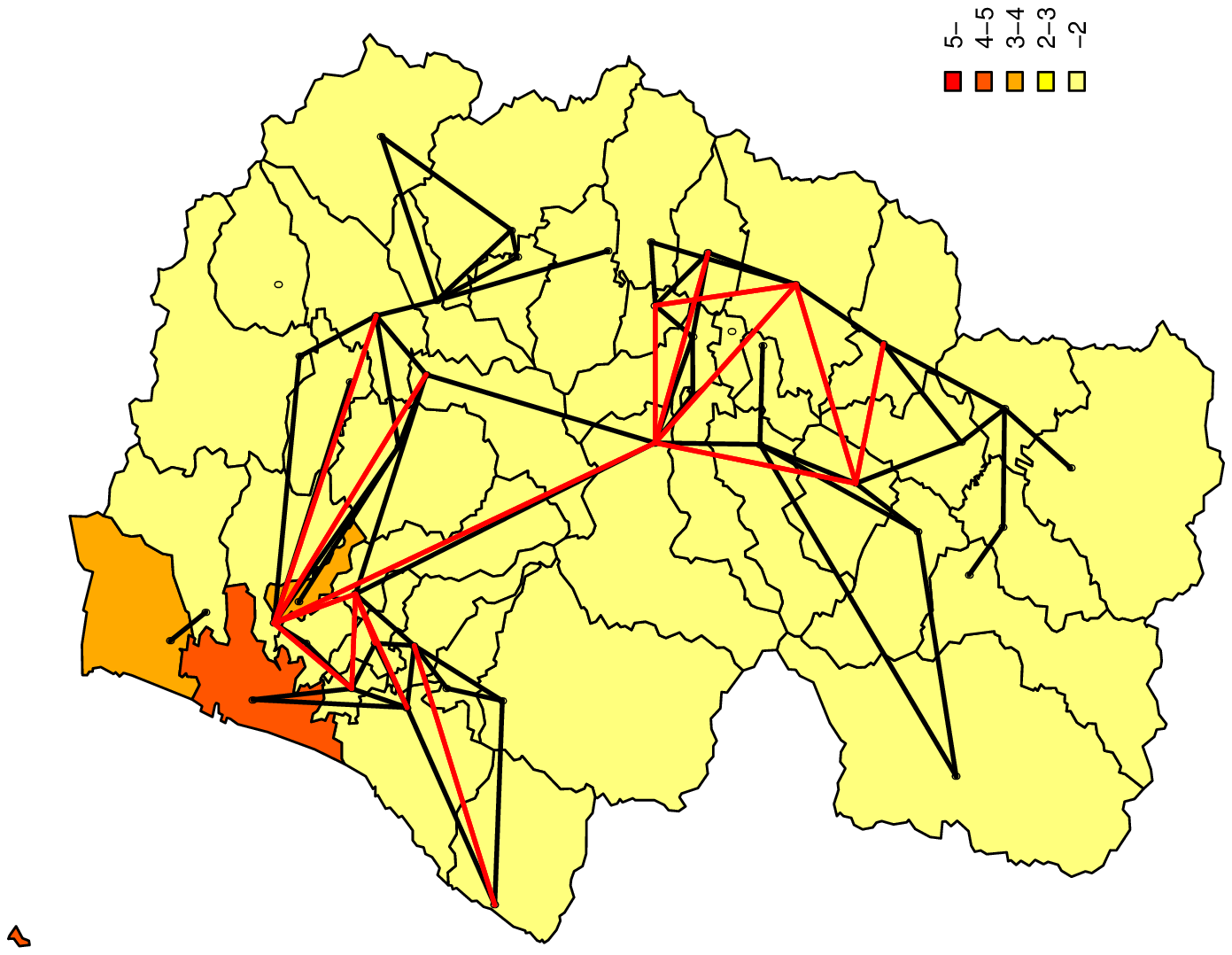} &
 \hspace*{-80mm}
 \includegraphics[angle=-90,scale=0.5]{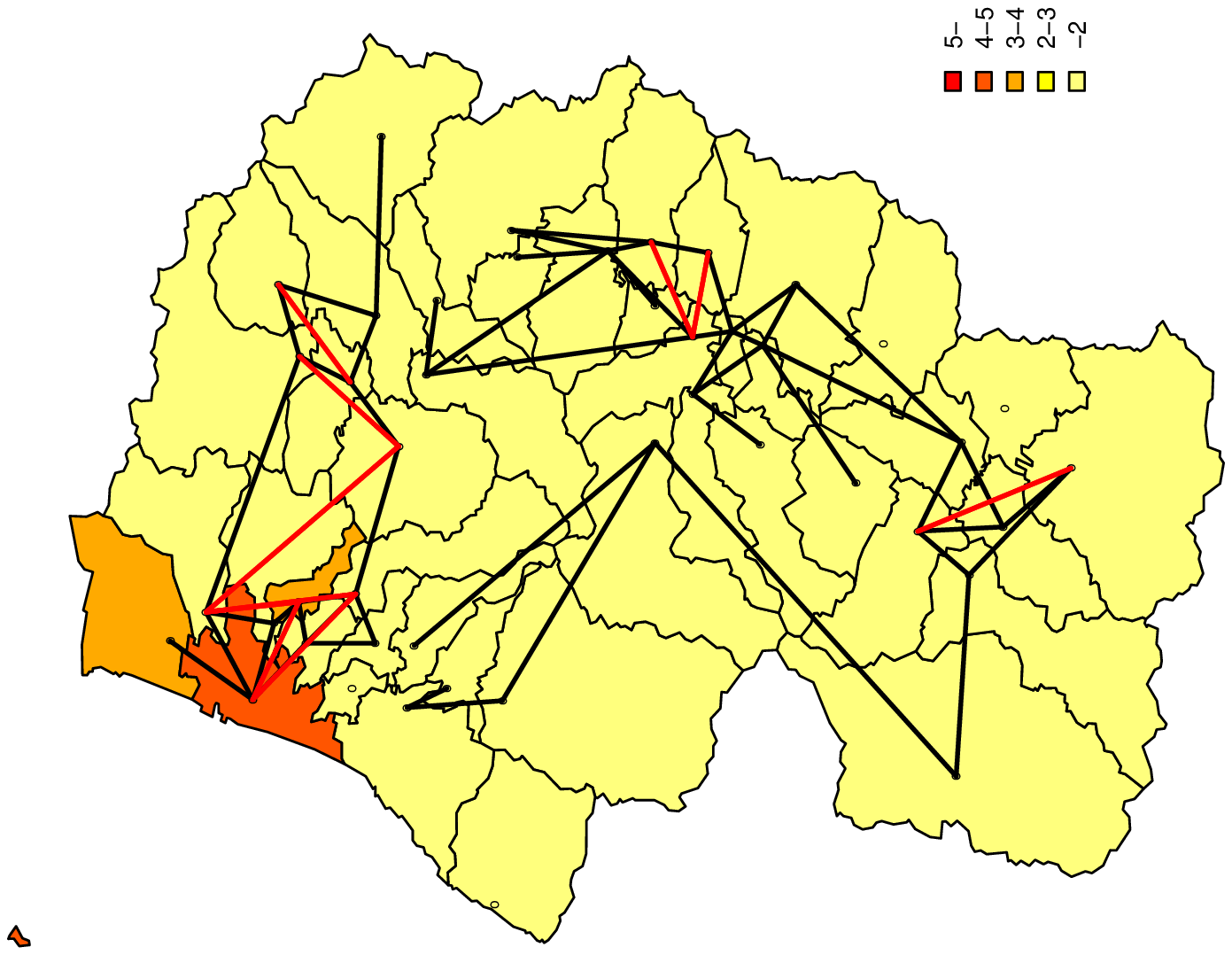}
\end{tabular}
\vspace{-10mm}
\caption{Generated graphs for the windows $|Z|\le 2$
 (left: group 1, right: group 2; circle: municipal capital, black line: edges of $G$, red line: additional edges for the chordal extension).}
\end{center}
\label{fig:part}
\end{figure}

\section{Summary and additional remarks}
\label{sec:summary}

In this paper, we proposed a recursive summation method to evaluate a class of expectation in multinomial distribution, and applied it to the evaluation of the $p$-values of temporal and spatial scan statistics.
This approach enabled us to evaluate the exact multiplicity-adjusted $p$-values.
Our proposal has an advantage where the true $p$-value is too small and is barely estimated precisely by Monte Carlo simulations.

The proposed algorithm is easily modified to a class of distributions including
the normal distribution, the Dirichlet distribution, the multivariate hypergeometric distribution, and the Dirichlet-multinomial distribution
by replacing recursive summations with recursive numerical integrations if necessary.

On the other hand, our proposed method has a limitation that it only works when the total data count $N$ and the window size are not large.
The limitation is due to $N$, $|B_i|$ (the sizes of the maximum clique), and $|k^{-1}(i)|$.
As shown in Section \ref{subsec:yamagata}, when the size of the scan windows is large, we can divide the whole scan windows into a number of groups, then compute the $p$-value for each group, and sum them.

We have implemented the proposed algorithms.
However, they are still under development.
Our algorithm requires {\tt for} loop calculations,
where the numbers of the nests and the ranges of running variables depend on the input (scan windows).
The dimensions of the arrays also depend on the data.
These make source coding complicated, and the resulting code inefficient.
One approach to overcome this difficulty may be the use of a preprocessor to generate the C program.
This approach has been successfully used in computer algebra (e.g., page 674 of \cite{Koyama-etal14}).

\subsubsection*{Acknowledgments}

The authors are grateful to Takashi Tsuchiya, Anthony J.\ Hayter, Nobuki Takayama, and Yi-Ching Yao for their helpful comments.
This work was supported by JSPS KAKENHI Grant Numbers 21500288 and 24500356.

\bibliographystyle{elsarticle-harv}

\end{document}